\documentclass[12pt]{article}

\usepackage{amsfonts}
\usepackage{amsmath}
\usepackage{amsthm}

\numberwithin{equation}{section}

\allowdisplaybreaks

\theoremstyle{plain}

\newtheorem{prop}{Propositon}[section]
\theoremstyle{definition}

\theoremstyle{remark}

\newcommand{\M}{{\cal M}}
\newcommand{\N}{{\cal N}}

\begin{document}
		\today

	\rightline{\baselineskip16pt\rm\vbox to20pt{
			{
				\hbox{OCU-PHYS-526}
				\hbox{AP-GR-163}
				\hbox{NITEP-84}
			}
			\vss}}%

\vspace{2cm}
	
\begin{center}

{\LARGE\bf 
Solutions to the Einstein-Maxwell-Current System with 
Sasakian maifolds}\\
\vspace{1cm}
	{\large 
		Hideki Ishihara$^a$\footnote{ishihara@sci.osaka-cu.ac.jp}
			~  and ~ 
		Satsuki Matsuno$^b$\footnote{smatsuno@sci.osaka-cu.ac.jp}
	}\\
	\bigskip
{\it 
$^{a,b}$Department of Mathematics and Physics, Graduate School of Science,\\
$^{a}$Nambu Yoichiro Institute of Theoretical and Experimental Physics (NITEP),\\
$^{a}$Osaka City University Advanced Mathematical Institute (OCAMI), \\
	Osaka City University,	3-3-138 Sugimoto, Sumiyoshi-ku, Osaka 558-8585, Japan}

\end{center}
\vspace{2cm}
\begin{abstract}
We construct stationary solutions to the Einstein-Maxwell-current system by using the 
Sasakian manifold for the three-dimensional space. 
Both the magnetic field and the electric current in the solution are specified 
by the contact form of the Sasakian manifold. 
The solutions contain an arbitrary function that describes 
inhomogeneity of the number density of the charged particles, 
and the function determines the curvature of the space. 
\end{abstract}

\newpage
\tableofcontents

\newpage
 
\section{Introduction}	

The origin of magnetic fields observed in various scales in the universe 
is one of the most important problem in astrophysics\cite{Widrow:2002ud}. 
The strong magnetic fields that cause energetic physical phenomena
become sources of gravitational field in the general relativity. 
It is an interesting task to clarify the structure of space-time with the magnetic field, 
but it is not easy to solve the coupled system of the Maxwell equation and 
the Einstein equation. 

The Bertotti-Robinson solutions~\cite{Bertotti:1959pf, Robinson:1959ev}
and the Melvin solutions~\cite{Melvin:1963qx} are well-known as the solutions to 
the Einstein-Maxwell system, where magnetic fields are source of the gravity. 
In these solutions, the magnetic fields are sourceless vacuum solutions to 
the Maxwell equation.  
In contrast, we present, in this paper, a family of solutions to the coupled system of 
the Einstein equation, the Maxwell equation and equation of motion of the electric current 
by using the Sasakian manifolds as a three-dimensional space. 

Recently, Sasakian manifolds \cite{Blair, Boyer} 
gather much attention in the context of M-theory~\cite	{Maldacena:1997re, Klebanov:1998hh}. 
On the other hand, magnetic fields on the Sasakian manifold and motion of charged 
particles in the magnetic fields 
are extensively studied~\cite{Cabrerizo et al 2009, Romaniuc et al 2015}. 
We present a new application of the Sasakian manifold in the Einstein-Maxwell-current 
system.

The Reeb vector, $\xi$, of the contact form, $\eta$, of a contact metric space 
is divergence-free, parallel to its rotation, and geodesic tangent~\cite{Blair}. 
We consider a direct product of time and a three-dimensional contact metric manifold 
as a four-dimensional space-time, 
and set a magnetic field and an electric current field that generates the magnetic field. 
We assume both the magnetic field and the current are characterized by $\xi$, 
and call them {\it contact magnetic field} and {\it contact current}, respectively. 
The contact current is assumed to be carried by a gas consisting many charged particles 
moving along $\xi$. 
The energy-momentum tensor of the coupled system of the gas and the contact magnetic 
field requires the three-dimensional space is $\eta$-Einstein. 
Therefore, the Sasakian manifold, which is a contact metric manifold and 
$\eta$-Einstein space~\cite{Blair}, together with the contact magnetic field 
and the contact current solves the Einstein-Maxwell-current system. 

We present stationary solutions to the coupled equations of the Einstein-Maxwell-current 
with a cosmological constant. The solutions have an arbitrary function that describes 
the number density of the charged particles,  
and the function determines the sectional curvature of the Sasakian manifold. 
It is shown that if the three-dimensional space of the solution is compact and simply connected, 
it is homeomorphic to three-dimensional sphere. 
In the limit of the particles at rest, the space-time reduces to Einstein's static 
universe, $\mathbb{R}\times S^3$.

The organization of this paper is as follows. 
In Section \ref{sec contact geometry}, 
we recall the contact metric geometry and the Sasakian manifold, 
and review some properties of the Reeb vector of the contact form. 
In Section \ref{sec solution}, 
we present stationary solutions with the Sasakian manifold 
to the Einstein-Maxwell-current system, 
and, in Section \ref{section examples}, we show solutions explicitly 
by using the Sasakian space forms, as examples. 
Section \ref{sec conclusion} is devoted to the summary.

\newpage

\section{Contact Metric Manifolds and Sasakian Manifolds}
	\label{sec contact geometry}
	We recall the definition of the contact metric manifold and some known facts. 
	A contact metric manifold is a 5-tuple $(\M,\phi,\xi,\eta,g)$ of a $(2n+1)$-dimensional Riemannian manifold 
	$(\M,g)$, a 1-form field $\eta$, a vector field $\xi$ and a $(1,1)$-type tensor field $\phi$ such that
\begin{align}
&\eta\wedge (d\eta)^n=\eta\wedge d\eta\wedge\cdots\wedge d\eta\ne0,\ \iota_\xi d\eta=0,\ \eta(\xi)=1,
	\\
&\phi^2=-I+\eta\otimes\xi,\ g(\phi(X),\phi(Y))=g(X,Y)-\eta(X)\eta(Y),\ g(X,\phi(Y))=d\eta(X,Y) 
	\label{d_eta}
\end{align}
holds.\footnote{
For 1-forms $\alpha^i$ and vector fields $X_j,\quad (i,j=1,2\cdots,k)$, 
we adopt the convention \\
$(\alpha^1\wedge\cdots\wedge\alpha^k)(X_1,\cdots,X_k):=\det(\alpha^i(X_j))/k!$.
}
The 1-form $\eta$ and the vector $\xi$ are called the contact form and the Reeb vector field, respectively. 
Through the metric $g$, $\eta$ and $\xi$ are related by $g(\xi, X)=\eta(X)$, or simply, 
$\xi={^\#\eta},~ \eta={^\flat\xi}$ by the musical isomorphism notation. 

A sectional curvature invariant under the action of $\phi$ is called a 
$\phi$-sectional curvature, and a sectional curvature of the plane spanned by $\xi$ 
and any vector field $X$ is called a $\xi$-sectional curvature. 
A contact metric manifold admits the $\phi$-sectional curvature.  

There are some equivalent definitions of the Sasakian manifold. 
We adopt the following definition: 
a contact metric manifold $(\M,\phi,\xi,\eta,g)$ is Sasakian if
\begin{align}
	R(X,Y)\xi=\eta(Y)X-\eta(X)Y 
	\label{eq sasakian condition}
\end{align}
holds, where $R$ is the Riemann curvature tensor of $\M$. 
It follows that the $\xi$-sectional curvature is 1.
A contact metric manifold $(\M,\phi,\xi,\eta,g)$ is called K-contact if $\xi$ is 
a Killing vector field. 
In general, a Sasakian manifold is K-contact, 
furthermore in three dimensions, a K-contact manifold is Sasakian. 
A Sasakian manifold that has a constant $\phi$-sectional curvature is called 
a Sasakian space form.

A contact structure of $\M$ is called regular if $\xi$ is a regular vector field, 
that is, an arbitrary point $p\in \M$ has a neighborhood $U$ such that the number of 
components of the intersection of $U$ and an integral curve of $\xi$ through 
any point $q\in U$ is one. 
We consider only compact and regular contact manifolds in this paper. 
In general, a compact regular contact manifold is a $S^1$ bundle on a compact 
symplectic manifold, and its fibers are integral curves of the Reeb vector field $\xi$~ 
(Boothby-Wang fibration). 
As a particular case, a Sasakian manifold is a $S^1$ bundle on a K\"{a}hler manifold. 
In more detail, let $(\M,\phi,\xi,\eta,g)$ be a compact regular Sasakian manifold, 
then there exists the Riemmanian submersion $\pi:(\M,g)\rightarrow (\N,h)$ 
whose fibers are integral curves of $\xi$,  
where $(\N,h)$ is the K\"{a}hler manifold with the metric $h$ and 
the K\"{a}hler form $\omega$ satisfying 
$d\eta=\pi^\ast\omega$.
	
\begin{prop}
		\label{lem Reeb geodesic}
		On a contact metric manifold $(\M,\phi,\xi,\eta,g)$, $\nabla_\xi\xi=0$ holds.
\end{prop}
\begin{proof}
	For $X\in \Gamma(T\M)$, 
		$$
		0=(\iota_\xi d\eta)(X)=(\mathcal{L}_\xi\eta)(X)=\xi g(\xi,X)-g(\xi,[\xi,X])=g(\nabla_\xi\xi,X), 
 		$$
		where we have used $\eta(\xi)=g(\xi,\xi)=1$. 
\end{proof}

\begin{prop}
	[\cite{Cabrerizo et al 2009}]
	\label{prop ast F}
	On a three-dimensional contact metric manifold $(\M,\phi,\xi,\eta,g)$, 
	$\ast d\eta=2\eta$ holds, where the symbol $*$ denotes the Hodge star operation 
	with respect to $g$.
\end{prop}

\begin{proof}
	Let $\Omega$ be the volume form of $\M$. 
For $X, Y, Z\in \Gamma(TM)$, we define the vector product 
	$X\times Y\in \Gamma(TM)$ by $g(X\times Y,Z)=\Omega(X,Y,Z)$. 
	Using the local orthonormal basis $\{\xi,e,\phi e \}$, called a $\phi$-basis, 
	we define an orientation of $\M$ such that $\Omega(\phi e,e,\xi)=1>0$. 
	Using the $\phi$-basis, we represent the volume form as 
	$\Omega(X,Y,Z)=\varepsilon_{ijk}X^iY^jZ^k$, where $\varepsilon_{ijk}$ is 
	the totally anti-symmetric tensor field. 
	
	For an arbitrary vector field $X\in\Gamma(TM)$, we have
\begin{align}
		\phi X=X\times\xi.\label{eq phi wedge in 3dim}
\end{align}
		Indeed, it is sufficient to show \eqref{eq phi wedge in 3dim} 
	for $\phi$-basis. 
	First, at any point on $\M$, $\phi e=a (e\times\xi)$ trivially holds, 
	where $a$ is a real number, 
	and $1=g(\phi e,\phi e)=ag(e\times\xi,\phi e)=a\Omega(e,\xi,\phi e)$, thus $a=1$.
	Second, we have $\phi(\phi e)=b(\phi e\times\xi)$,
 	where $b$ is a real number, 
	and $-1=g(e,\phi^2e)=bg(e,\phi e\times\xi)=b\Omega(\phi e,\xi,e)=-b$ , then $b=1$. 
	It is apparent that \eqref{eq phi wedge in 3dim} holds for $X=\xi$. 
	Therefore, \eqref{eq phi wedge in 3dim} holds for an arbitrary $X$.	

	From \eqref{d_eta}, we have 
\begin{align}
		d\eta(X,Y)=g(X,\phi Y)=g(X,Y\times \xi) 
		=\Omega(X,Y,\xi)
		=\varepsilon_{ijk}X^i Y^j \xi^k=2 (*\eta)(X,Y), 
\end{align}
	then $*d\eta=2\eta$. 
\end{proof}
From Proposition \ref{prop ast F}, on a three-dimensional contact metric manifold $(\M,\phi,\xi,\eta,g)$, 
we have
\begin{align}
	&{\rm div} ~\xi :=\ast d {(\ast \eta)}= \frac12~{*d(d\eta)}=0, \label{div_free}
	\\
	&{\rm rot} ~\xi :={^\#(\ast d \eta)} = 2~ {^\#\eta}= 2\xi, \label{Beltrami}
\end{align}
namely, the Reeb vector $\xi$ is divergence-free, and is parallel to 
its own rotation\footnote{A vector field with this property is known as the 
Beltrami field\cite{Beltrami, Yoshida_Giga, Reed}. 
}.

\section{Solutions to the Einstein-Maxwell-Current System}
	\label{sec solution}
\subsection{Equations of the system}
Let $(\M,\phi,\xi,\eta,g)$ be a three-dimensional contact metric manifold, and 
we consider a four-dimensional space-time 
$(\widetilde{\M}=\mathbb{R}\times \M,\tilde{g}=-dt^2+g)$, where $g$ is independent on $t$.
	We extend  $\eta,\xi,\phi$ 
	to the fields on $\widetilde{\M}$ by the following way; 
	that is, let $p$ be the Riemannian submersion $p:\widetilde{\M}=\mathbb{R}\times \M\rightarrow \M$, 
	then $\widetilde{\eta}:=p^\ast\eta$ is the 1-form field on $\widetilde{\M}$.
	For the docomposition $T_x\widetilde{\M}=\mathbb{R}\oplus T_{p(x)}M$, we extend $\xi$ 
	to the vector field $\tilde{\xi}:=(0,\xi)$, 
	and we define the $(1,1)$-type tensor field $\widetilde{\phi}$ on $\widetilde{\M}$ 
	by $\widetilde{\phi}(X):=(0,\phi( p_\ast X))$. 
	Hereafter, we drop \lq tilde\rq\ and use $\eta,\xi,\phi$ for the four-dimensional 
	quantities on $\widetilde{\M}$, for simplicity.  

We consider a Maxwell field $F=dA$ where $A$ is a potential form, and  
a collision-less many-particle system that carries an electric current $j$. 
We construct solutions to the Einstein-Maxwell-current system 
with a cosmological constant governed by the coupled equations 
\begin{align}
	&\widetilde{Ric}-\frac{1}{2}\widetilde{R}\tilde{g}+\Lambda\tilde{g}=T, 
		\label{Einstein_eq} \\
	&{\ast d{\ast F}}=j, 
		\label{Maxwell_eq} \\
	&m_{(i)} \nabla_{u_{(i)}} u_{(i)}=f_{(i)}, 
		\label{EOM_particle}
\end{align}
where $\widetilde{Ric},\widetilde{R}$ are the Ricci tensor and the scalar curvature of 
$(\widetilde{\M},\tilde{g})$, 
$T$ is the energy momentum tensor of the coupled system of the Maxwell field 
and the many-particle system, 
and $\ast$ is the Hodge star operator with respect to $\tilde{g}$. 
The equation \eqref{EOM_particle} denotes the equation of motion for a particle, 
labeled by $i$, where $m_{(i)}$ and $u_{(i)}$ are mass and 4-velocity of the particle. 
The right hand side of \eqref{EOM_particle} is the Lorentz force acting 
on the charged particle written by
\begin{align}
	f=-e ~{^\#(\iota_{u}F)}, 
\end{align}
where $e$ is the electric charge of the particle. 
The current $j$ carried by the particles is defined by 
\begin{align}
	{^\#j} = \sum_{i} e_{(i)}u_{(i)}. 
	\label{current}
\end{align}
Since the space-time geometry is assumed to be static, we consider the Maxwell field 
and all quantities that characterize the particles: number densities, 4-velocities, and so on 
are time independent.

\subsection{Contact magnetic field and contact current}
\label{contact_magnetic_field}

Using the contact form as a vector potential $A=B\eta$, 
we define the contact magnetic field \cite{Cabrerizo et al 2009, Romaniuc et al 2015} 
as 
\begin{align}
	F_B:=B d\eta, \label{contact_F}
\end{align}
where $B$ is a constant. 
For the Maxwell 2-form field $F$ on $\widetilde{\M}$, 
the electric vector field, $\overrightarrow{E}$, and magnetic vector field, 
	$\overrightarrow{B}$, are defined by 
\begin{align}
		g(\overrightarrow{E}, X):=F(X, \partial_t), \quad
 		g(\overrightarrow{B}, X):=\ast F(\partial_t,X). 
\end{align}
Then, for the contact magnetic field, 
we see $\overrightarrow{B}=2B\xi$ from Proposition \ref{prop ast F}, 
and $\overrightarrow{E}=0$. Immediately, it holds that ${\rm div}\overrightarrow{B}=0$. 
From \eqref{d_eta}, $\phi$ is related to the contact magnetic field by
\begin{align}
	F_B(X,Y) =g(X,B\phi(Y)). 
\end{align}
We also have $\ast d{\ast F_B}=4B\eta$, then the Maxwell equation \eqref{Maxwell_eq} requires 
\begin{align}
	j=4B\eta. 
	\label{eq contact current}
\end{align}
We call the current $j$ contact current. 
It should be noted that the contact magnetic vector field, $\overrightarrow{B}$, is 
parallel to the contact current, ${^\#j}$, that generates the magnetic field. 
Therefore, the  Lorentz force dose not act on the contact current\footnote{
The case that the magnetic vector field is parallel to the current is 
studied in various setups\cite{Reed, Marsh, Mahajan_Yoshida}.}. 
From Proposition \ref{lem Reeb geodesic} on $\M$, we also have $\nabla_\xi\xi=0$ 
on $\widetilde{\M}$, where $\nabla$ is the Riemman connection with respect to $\tilde{g}$.
Then, the contact current flows along the geodesic curves. 

In the case that the magnetic vector field is not parallel to the current field that generates 
the magnetic field, usually this case occurs, 
the Lorentz force by the magnetic field acts on the current. 
It would be a complicated task to find a consistent solution 
to the equations of motion of the current and the Maxwell equation 
if the back-reaction on the current field from the magnetic field is 
taken into account. 
In contrast, the contact current that generates the contact magnetic field flows 
along geodesic curves without the back-reaction. 

\subsection{Gas model}

We consider a gas as the many-particle system. 
The gas is assumed to contain a huge number of particles, then 
the 4-velocity $u$ is treated as a smooth vector field on $\widetilde \M$. 
We assume the gas consists of three-components of particles: 
particles with mass $m$ and charge $e=+1$ , particles with mass $m$ and charge $e=-1$, 
and neutral particles with mass $m_n$. 
The collision of particles and the thermal motion are neglected. 

We set the number densities of charged particles, $n_+$ and $n_-$ take 
the same value $n\in C^\infty(\M)$, so that the gas is electrically neutral, 
and the 4-velocities of the charged particles are 
\begin{align}
	u_\pm=u^0\partial_t \pm u^3 \xi, 
	\label{4-velocity}
\end{align}
where $u_\pm$ are unit timelike vector fields, i.e., 
\begin{align}
	-(u^0)^2+(u^3)^2=-1. 
	\label{eq u constraint}
\end{align}
The number density of the neutral particles is denoted by $n_n\in C^\infty(\M)$, 
and the neutral particles are assumed to be at rest, namely the 4-velocity is $u_n=\partial_t$. 

The charged particles move along the contact magnetic vector field, 
then the Lorentz force acting on the particles vanishes. 
Therefore, all particles move along geodesics, 
\begin{align}
	\nabla_{u_\pm} u_\pm=0 ~\text{and}~ \nabla_{u_n} u_n=0. 
\label{geodesics}
\end{align}
Since $\partial_t$ and $\xi$ are geodesic tangents, we further assume that
\begin{align}
	\xi(u^0)=0,~\text{and}~ \xi(u^3)=0, \label{eq u0 u3 independent on xi}
\end{align}
so that \eqref{geodesics} holds. 

The electric current \eqref{current} carried by the charged particles 
with the 4-velocity \eqref{4-velocity} is given by 
\begin{align}
	{^\# j}=nu_+-nu_-=2nu^3\xi,
\end{align}
and the Maxwell equation \eqref{eq contact current} requires 
\begin{align}
	n u^3=2B. 
	\label{u3_by_B}
\end{align}
The contact magnetic field and the gas as the many-particle system that carry the contact current 
solve the coupled system of the Maxwell equation and the equations of motion of the particles. 
	
\subsection{Energy-momentum tensors}

The energy momentum tensor of a Maxwell field $F$ defined by\footnote{
The symbol $\tilde{g}^*$ gives the inner product of $\Omega^1(M)$, and the inner product 
for $\Omega^p(M)$ is defined by 
$\tilde{g}^*(\alpha^1\wedge\cdots\wedge\alpha^k,\beta^1\wedge\cdots\wedge\beta^k)
=\det(\tilde{g}^*(\alpha^i,\beta^j))/p!$.
}
\begin{align}
	T_{EM}(X,Y)=\tilde{g}^*(\iota_X F,\iota_Y F)-\frac{1}{4}||F||^2 \tilde{g}(X,Y)  	
\end{align}
reduces to 
	\begin{align}
		T_{EM}=\frac{1}{2}B^2(dt\otimes dt+g-2\eta\otimes \eta)
		\label{T_EM}
	\end{align}

	for the contact magnetic field \eqref{contact_F}. 

The energy momentum tensor of the gas that consists of charged particles and 
neutral particles is given by
\begin{align}
	T_F =nm\left(({^\flat u_+}) \otimes ({^\flat u_+}) 
			+ ({^\flat u_-})\otimes ({^\flat u_-})\right)
			+n_n m_n ({^\flat u_n})\otimes ({^\flat u_n}), 
	\label{T_gas}
\end{align}
where 
\begin{align}
	{^\flat u_\pm}=- u^0 dt \pm u^3\eta, \quad \text{and}\quad
	{^\flat u_n}= - dt.  \label{u_a}
\end{align}
Insert \eqref{u_a} into \eqref{T_gas}, we have
\begin{align}
	T_F =\rho ~dt\otimes dt+P~\eta\otimes\eta, 
	\label{T_F}
\end{align}
where $\rho$ is the energy density of the gas and $P$ is the effective pressure 
caused by the motion of particles. These are given by
\begin{align}
	&\rho =2nm(u^0)^2+n_nm_n, \label{rho}\\
	&P =2nm(u^3)^2. \label{P}
\end{align}
It is apparently that $\partial_t \rho=0$ and $\partial_t P=0$.
%

\subsection{Solutions}
	We inspect the condition that $(\widetilde{\M},\tilde{g})$ solves 
	the Einstein-Maxwell-current system with the contact current 
	and the contact magnetic field as sources of the gravitational field.
	The Einstein equation \eqref{Einstein_eq} is rewritten in the form 
	\begin{align}
		\widetilde{Ric} = T -\frac{1}{2} {\rm tr} T ~\tilde{g}+\Lambda \tilde g, 
		\label{Enstein_eq_2}
	\end{align}
	where $T:=T_{EM}+T_F$.
We decompose this equation into the time component and the space components of $\M$, 
and using \eqref{T_EM} and \eqref{T_F}, we obtain
\begin{align}
	&\widetilde{Ric}(\partial_t,\partial_t)=0
		=\frac12\left(\rho+P+B^2\right)-\Lambda,
		\label{eq einsten time comp}
		\\
	&Ric=\frac12( \rho-P + B^2 +2\Lambda)~g+(P-B^2)~\eta\otimes \eta, 
		\label{eq einsten sp comp}
\end{align}
where $Ric$ denotes the Ricci tensor of $(\M, g)$. 
Equation \eqref{eq einsten sp comp} shows 
the Ricci tensor of the contact metric manifold $(\M,\phi,\xi,\eta,g)$ 
has the form of $Ric=ag+b\eta\otimes\eta,\ a,b\in C^\infty(\M)$. 
Therefore, the manifold should be an $\eta$-Einstein. 

The space $(\M,g)$ being $\eta$-Einstein requires the Ricci tensor are diagonal 
in the $\phi$-basis, $\{e,\phi e,\xi \}$, on $\M$. 
It is easy to derive the diagonal components of the Ricci tensor in the forms: 
\begin{align}
	&Ric(e,e)=K(e,\xi)+ K(e,\phi e), \\
	&Ric(\phi e,\phi e)=K(\phi e,\xi)+ K(e,\phi e), 	
	\label{Ricci_by_K} \\
	&Ric(\xi,\xi)=K(e,\xi)+K(\phi e,\xi), 
\end{align}
	where $K$ denotes the sectional curvature. 
	As mentioned before, a three-dimensional contact metric manifold $\M$ admits a 
	$\phi$-sectional curvature $H =K(e,\phi e)$,  
	and furthermore, if the manifold $\M$ is $\eta$-Einstein, it is known that 
	there are three possibilities for $\M$: 
	(a) a contact metric manifold that has constant $\xi$-sectional curvature, say $k$, less than 1, and $H=-k$, 
	(b) a three-dimensional flat contact metric manifold, and 
	(c) a Sasakian manifold~\cite{Blair et al 1990}. 

In the case of (a), 
substitute $K(\xi,e)=K(\xi,\phi e)=k$ and $K(e,\phi e)=H=-k$ into \eqref{Ricci_by_K}, 
the Ricci tensor reduces to $Ric=2k~\eta\otimes\eta$. 
The Einstein equations \eqref{eq einsten time comp} and \eqref{eq einsten sp comp} 
lead to $\rho=-B^2$. 
We exclude the case (a) because the negative energy is unphysical. 
In the case of (b), 
since $Ric=0$, we have $\rho=-B^2$ from \eqref{eq einsten time comp} 
and \eqref{eq einsten sp comp} as same as the case (a). 
Then we also exclude the case (b). 

In the case of (c), since $(\M,\phi,\xi,\eta,g)$ is a three-dimensional Sasakian manifold, 
the $\xi$-sectional curvature $K(\xi,e)=K(\xi,\phi e)=1$. 
Then, the Ricci tensor is given by 
\begin{align}
	Ric=(1+H)g+(1-H)\eta\otimes\eta. 
\end{align}
Then, the Einstein equations \eqref{eq einsten time comp} and \eqref{eq einsten sp comp} 
become 
\begin{align}
	&\Lambda=1+\frac{B^2}{2}, 
		\label{eq einstein lambda} \\
	&\rho=1-B^2+H,
 		\label{eq einstein rho}\\
	&P+\rho=2. 
		\label{eq einstein 2} 
\end{align}
The equation \eqref{eq einstein lambda} means that the cosmological constant is positive. 
Since the Reeb vector $\xi$ is a Killing vector in the Sasakian manifold, 
$\xi(\rho)=0$ and $\xi(P)=0$ hold. 

Equations \eqref{eq einstein rho} and \eqref{eq einstein 2} with the help of 
\eqref{eq u constraint}, \eqref{u3_by_B}, \eqref{rho} and \eqref{P} lead to 
\begin{align}
	&H = 1 + B^2 - 8  \frac{B^2 m}{n},	\label{eq relation H and n}\\
	&m_n n_n = 2 - 16 \frac{B^2 m}{n} - 2 m n. 
\end{align}
The functions $n, n_n$ should satisfy $\xi(n)=\xi(n_n)=0$, 
then these are functions on the base space $\N$ of Boothby-Wang fibration 
for the Sasakian manifold $\M$. 
The functions $H$ and $n_n$ are determined by the function $n$. 
Because the number density of the particles should be positive, 
we see that $B$ and $mn$ should be in the range 
\begin{align}
	0<(Bm)^2<\frac{1}{32},\quad \text{and}\quad
	1 - \sqrt{1 - 32 B^2 m^2} < 2mn < 1+ \sqrt{1 - 32 B^2 m^2}. 
	\label{range_n}
\end{align}
The function $H$ is expressed by
\begin{align}
		H = B^2+m n+ \frac12 m_n n_n > 0, 
\end{align}
and it can vary on $\M$ in the range
\begin{align}
	\frac{1}{2}+B^2-\frac{1}{2}\sqrt{1-32B^2m^2}<H<\frac{1}{2}+B^2+\frac{1}{2}\sqrt{1-32B^2m^2}. 
	\label{eq H range}
\end{align}


Here, we specify the class of Sasakian manifolds of the solution.
We concentrate on a compact, simply connected and regular Sasakian manifold. 
Namely, the manifold $\M$ is 
a $S^1$ bundle on a compact symplectic manifold, $\N$, 
and its fibers are integral curves of the Reeb vector field $\xi$. 
It is known that the $\phi$-sectional curvature of $\M$ and the sectional curvature 
$K_\ast$ of $\N$ are related as $H=K_\ast-3$. Since we have $H>0$, then $K_\ast>3$. 
Any two-dimensional complete Riemannian manifolds admitting everywhere positive 
sectional curvature is homeomorphic to $S^2$, then $\N\simeq S^2$.
	
Therefore, $(\M,g)$ is a $S^1$ bundle on a K\"{a}hler manifold $(S^2,h,\omega)$, 
where Riemannian metric $h$ satisfies the condition \eqref{eq H range} 
for its sectional curvature $K_\ast=H+3$, and $\omega$ is a K\"{a}hler form.
We can take locally a 1-form $\tau$ such that $d\tau=\omega$, and let $z$ be 
a fiber coordinate, then the metric $g$ is given by 
\begin{align}
	g&=\pi^\ast h+\eta\otimes\eta,\\
	\eta&=dz+\tau.
\end{align}

We can regard $\eta$ as a connection form of the $S^1$-bundle, and we have
\begin{align}
	\int_{S^2} d\eta = \int_{S^2}\omega \neq 0,
\end{align}
thus the Euler number of $\M$ is not zero. Then, $\M$ is not a direct product bundle. 
Since $\M$ is assumed to be simply connected, there exists a bundle isomorphism from $\M$ 
to the Hopf fibration. Therefore, $\M$ with this metric is homeomorphic to $S^3$.

	\section{Examples of Solutions}\label{section examples}
	A compact and simply connected three-dimensional Sasakian space form 
	with a positive constant $\phi$-sectional curvature is homeomorphic to $S^3$.
	Then the metric can be written in the form of a $S^1$ bundle on $S^2$: 
	\begin{align}
		g=\frac{\alpha}{4}(d\theta^2+\sin^2\theta d\phi^2)+\frac{\alpha^2}{4}(d\psi+\cos\theta d\phi)^2,
	\label{eq S3 D-defm metric}
	\end{align}
where $\alpha$ is a positive constant, 
	and $\psi$ is a fiber coordinate. 
The contact form and the Reeb vector are 
\begin{align}
	\eta=\frac{\alpha}{2}(d\psi+\cos\theta d\phi), \quad \xi =\frac{2}{\alpha}\partial_\psi, 
\end{align}
respectively. 

The space-time is homeomorphic to $\mathbb{R}\times S^3$ and the metric is 
\begin{align}
	\tilde{g}=-dt^2 +g, 
	\label{eq sasaki space form solution}
\end{align}
and the constant $\phi$-sectional curvature is given by $H=\frac{4}{\alpha}-3$. 
The contact magnetic field $F_B$ and the contact current $j$ are
\begin{align}
	&F_B=-\frac{\alpha B}{2}\sin\theta d\theta\wedge d\phi,\\
	&j=\frac{8B}{\alpha}\partial_\psi.
\end{align}
	
From \eqref{eq relation H and n}, the parameter $\alpha$ is expressed by 
the number density of the charged particles, 
a constant on $\widetilde\M$ in this case,  
in the form 
\begin{align}
	\frac{1}{\alpha} = 1+\frac{B^2}{4}\left(1-\frac{8m}{n}\right). 
\end{align}
In order that $H$ satisfies \eqref{eq H range}, 
there are varieties of solutions with $\alpha$ in the range
	\begin{align}
		\frac{14 + 4 B^2-2\sqrt{1-32B^2m^2}}{12 + B^4 + B^2 (7 + 8 m^2)}
		< \alpha
		< \frac{14 + 4 B^2+2\sqrt{1-32B^2m^2}}{12 + B^4 + B^2 (7 + 8 m^2)}. 
	\end{align}

If we set $u^3=0$, we have $j=0, B=0, \rho=2, P=0, \Lambda=1$, and $\alpha=1$. 
This is Einstein's static universe. 
Therefore, the solutions based on the Sasakian space forms \eqref{eq sasaki space form solution} 
are generalizations of Einstein's static universe by the existence of the contact magnetic field 
and the gas that carries the contact electric current. 
	

\section{Summary}
\label{sec conclusion}

We constructed exact stationary solutions to the Einstein-Maxwell-current system 
by using a direct product of time and a three-dimensional Sasakian manifold. 
We considered that a gas of collision-less charged particles that carries stationary 
electric current, and the current generates a magnetic field. 
Taking both the gas and the magnetic field 
as the source of gravitational field, we have presented solutions 
to the coupled system of the equations of particles in the gas, the Maxwell equation 
and the Einstein equation with a positive cosmological constant. 

The three-dimensional space was assumed to be a contact metric manifold, 
which has the contact form, $\eta$, and the Reeb vector, $\xi$. 
The key properties are the followings: $\xi$ is divergence-free, 
parallel to its own rotation, and geodesic tangent. 
%
We considered that both the electric current and the magnetic vector field are parallel to $\xi$. 
%

The electric current is assumed to be carried by moving charged particles in a gas . 
Through the Einstein equation, 
the energy momentum tensor of the fluid and the Maxwell field generated by the current 
requires the three-dimensional contact metric manifold to be a Sasakian manifold, 
which has $\eta$-Einstein metric. 
As the result, the Einstein equations are reduced to a coupled system of algebraic equations 
that is easy to solve.

The solution contains a function that describes inhomogeneity of the number density 
of particles. 
Accordingly, $\phi$-sectional curvature of the Sasakian manifold of the solution is 
positive and can vary in the range 
such that the number density of the particles is positive.
It was also shown that if three-dimensional space is compact and simply connected, 
the space-time topology of the solution is $\mathbb{R}\times S^3$.

We presented explicitly the solution with anisotropy of the three-dimensional space 
by using the Sasakian space form as an example. 
In the case that the particles in the gas are at rest, 
the solution reduces to Einstein's static universe. 
Then, the solutions obtained in the present paper are generalizations of it 
by introducing the gas that carries the electric current and magnetic field that 
are characterized by the contact form. 

\section{Acknowledgement}
This work was partly supported by Osaka City University Advanced 
Mathematical Institute (MEXT Joint Usage/Research Center on Mathematics 
and Theoretical Physics JPMXP0619217849).


\end{document}